\documentclass[copyright,creativecommons]{eptcs}
\providecommand{\event}{AFL 2014}

\usepackage[utf8]{inputenc}
\usepackage[T1]{fontenc}
\usepackage{makeidx,amsmath,amssymb,dsfont,algorithm,amsthm,pifont}
\usepackage{paralist,url,color}
\usepackage[noend]{algorithmic}
\usepackage[only,llbracket,rrbracket]{stmaryrd}
\usepackage[strings]{underscore}

\allowdisplaybreaks
\algsetup{indent=2em}

\newtheorem{definition}{Definition}

\newtheorem{lemma}[definition]{Lemma}
\newtheorem{proposition}[definition]{Proposition}
\newtheorem{theorem}[definition]{Theorem}

\DeclareMathOperator{\merge}{merge}
\DeclareMathOperator{\sg}{Sig}
\DeclareMathOperator{\rk}{rk}
\DeclareMathOperator{\height}{ht}
\DeclareMathOperator{\pos}{pos}
\DeclareMathOperator{\supp}{supp}
\DeclareMathOperator{\wt}{wt}

\providecommand*{\suc}[0]{\ensuremath{\mathrm{succ}}}
\providecommand*{\abs}[1]{\ensuremath{\lvert #1 \rvert}}
\providecommand*{\seq}[3]{\ensuremath{#1_{#2}, \dotsc, #1_{#3}}}

\providecommand*{\symdiff}[0]{\ensuremath{\mathop{\ominus}}}
\providecommand*{\nat}[0]{\ensuremath{\mathbb N}}
\providecommand*{\semf}[1]{\ensuremath{\llbracket {#1} \rrbracket}}
\providecommand*{\sem}[1]{\ensuremath{\mathop{\semf{#1}}}}

\providecommand*{\bigo}{\mathcal{O}}
\providecommand*{\dta}{\textsc{dta}}
\providecommand*{\wdta}{\textsc{dwta}}
\providecommand*{\tup}[1]{\langle #1 \rangle}

\providecommand*{\sr}[1]{\mathds{#1}}

\providecommand*{\cokernel}{{\overline{K}}}
\providecommand*{\copreamble}{{\overline{P}}}
\providecommand*{\srplus}{+}
\providecommand*{\srtimes}{\cdot}

\providecommand*{\SBox}[0]{\ensuremath{\mathchoice%
    {{\scriptstyle \Box}}%
    {{\scriptstyle \Box}}%
    {{\scriptscriptstyle \Box}}%
    {{\scriptscriptstyle \Box}}%
}}

\title{Hyper-Minimization for \\ Deterministic Weighted Tree Automata}
\def\titlerunning{Hyper-minimization for deterministic weighted tree
  automata}

\author{Andreas Maletti\footnotemark[1] 
  \institute{Universit\"at Leipzig, Institute of Computer Science \\
  Augustusplatz~10--11, 04109 Leipzig, Germany}
  \email{\texttt{\upshape maletti@informatik.uni-leipzig.de}}
  \and
  Daniel Quernheim\footnotemark[1]
  \institute{Universit\"at Stuttgart, Institute for Natural Language
    Processing \\ Pfaffenwaldring~5b, 70569 Stuttgart, Germany}
  \email{\texttt{\upshape daniel@ims.uni-stuttgart.de}}}
\def\authorrunning{A.~Maletti and D.~Quernheim}

\def\copyrightholders{\authorrunning}

\begin{document}
\renewcommand{\thefootnote}{\fnsymbol{footnote}}
\footnotetext[1]{Both authors were financially supported by the German
  Research Foundation~(DFG) grant MA\,/\,4959\,/\,1--1.}
\renewcommand{\thefootnote}{\arabic{footnote}}

\maketitle

\begin{abstract}
  Hyper-minimization is a state reduction technique that allows a
  finite change in the semantics.  The theory for hyper-minimization
  of deterministic weighted tree automata is provided.  The presence
  of weights slightly complicates the situation in comparison to the
  unweighted case.  In addition, the first hyper-minimization
  algorithm for deterministic weighted tree automata, weighted over
  commutative semifields, is provided together with some
  implementation remarks that enable an efficient implementation.  In
  fact, the same run-time~$\mathcal O(m \log n)$ as in the unweighted
  case is obtained, where $m$~is the size of the deterministic
  weighted tree automaton and $n$~is its number of states.
\end{abstract}

\section{Introduction}
\label{sec:Intro}
Deterministic finite-state tree automata
(\dta)~\cite{gecste84,gecste97} are one of the oldest, simplest, but
most useful devices in computer science representing structure.  They
have wide-spread applications in linguistic analysis and
parsing~\cite{petbarthikle06} because they naturally can represent
derivation trees of a context-free grammar.  Due to the size of the
natural language lexicons and processes like state-splitting, we often
obtain huge \dta\ consisting of several million states.  Fortunately,
each \dta\ allows us to efficiently compute a unique (up to
isomorphism) equivalent minimal \dta, which is an operation that most
tree automata toolkits naturally implement.  The asymptotically most
efficient minimization algorithms are based on~\cite{koz92,hogmalmay08},
which in turn are based on the corresponding procedures for
deterministic string automata~\cite{hop71,gri73,valleh08}.  In
general, all those procedures compute the equivalent states and merge
them in time~$\bigo(m \log n)$, where $n$~is the number of states of
the input \dta\ and $m$~is its size.

Hyper-minimization~\cite{badgefshi07} is a state reduction technique
that can reduce beyond the classical minimal device because it allows
a finite change in the semantics (or a finite number of errors).  It
was already successfully applied to a variety of devices such as
deterministic finite-state automata~\cite{gawjez09,holmal10},
deterministic tree automata~\cite{jezmal12b} as well as deterministic
weighted automata~\cite{malque11b}.  With recent progress in the area
of minimization for weighted deterministic tree
automata~\cite{malque11c}, which provides the basis for this
contribution, we revisit hyper-minimization for weighted deterministic
tree automata.  The asymptotically fastest hyper-minimization
algorithms~\cite{gawjez09,holmal10} for \textsc{dfa} compute the
``almost-equivalence'' relation and merge states with finite left
language, called preamble states, according to it in time~$\bigo(m
\log n)$, where $m$~is the size of the input device and $n$~is the
number of its states.  Naturally, this complexity is the goal for our
investigation as well.  Variations such as cover automata
minimization~\cite{camsanyu01}, which has been explored before
hyper-minimization due to its usefulness in compressing finite
languages, or $k$-minimization~\cite{gawjez09} restrict the length of
the error strings instead of their number, but can also be achieved
within the stated time-bound.

As in~\cite{malque11b} our weight structures will be commutative
semifields, which are commutative semirings~\cite{hebwei98,gol99} with
multiplicative inverses.  As before, we will restrict our attention to
deterministic automata.  Actually, the mentioned applications of \dta\
often use the weighted version to compute a quantitative answer (i.e.,
the numerically best-scoring parse, etc).  We already know that
weighted deterministic tree automata~(\wdta)~\cite{berreu82,fulvog09}
over semifields can be efficiently minimized~\cite{malque11c},
although the minimal equivalent \wdta\ is no longer unique due to the
ability to ``push'' weights~\cite{moh97,eis03,malque11c}.  The
asymptotically fastest minimization algorithm~\cite{malque11c}
nevertheless still runs in time~$\bigo(m \log n)$.  To the authors'
knowledge, \cite{malque11c}~is currently the only published algorithm
achieving this complexity for \wdta.  Essentially, it normalizes the
input \wdta\ by ``pushing'' weights, which yields that, in the
process, the signatures of equivalent states become equivalent, so
that a classical unweighted minimization can then perform the
computation of the equivalence and the merges.  To this end, it is
important that the signature ignores states that can only recognize
finitely many contexts, which are called co-preamble states, to avoid
computing a wrong ``pushing'' weight.

We focus on an almost-equivalence notion that allows the recognized
weighted tree languages to differ (in weight) for finitely many trees.
Thus, we join the results on unweighted hyper-minimization for
\dta~\cite{jezmal12b} and weighted hyper-minimization for
\textsc{wdfa}~\cite{malque11b}.  Our algorithms (see
Algorithms~\ref{alg:Overall} and \ref{alg:Almost}) contain features of
both of their predecessors and are asymptotically as efficient as them
because they also run in time~$\bigo(m \log n)$.  As in~\cite{que10},
albeit in a slightly different format, we use standardized signatures
to avoid the explicit pushing of weights that was successful
in~\cite{malque11c}.  This adjustment allows us to mold our weighted
hyper-minimization algorithm into the structure of the unweighted
algorithm~\cite{holmal10}.

\section{Preliminaries}
\label{sec.prel}
We use~$\nat$ to denote the set of all nonnegative integers
(including~$0$).  For every integer $n \in \nat$, we use the set $[n]
= \{i \in \nat \mid 1 \leq i \leq n\}$.  Given two sets $S$~and~$T$,
their \emph{symmetric difference}~$S \symdiff T$ is given by $S
\symdiff T = (S - T) \cup (T - S)$.  An alphabet~$\Sigma$ is simply a
finite set of symbols, and a \emph{ranked alphabet}~$(\Sigma,
\mathord{\rk})$ consists of an alphabet~$\Sigma$ and a
ranking~$\mathord{\rk} \colon \Sigma \to \nat$.  We let $\Sigma_n = \{
\sigma \in \Sigma \mid \rk(\sigma) = n\}$ be the set of symbols of
rank~$n$ for every $n \in \nat$.  We often represent the ranked
alphabet~$(\Sigma, \mathord{\rk})$ by~$\Sigma$ alone and assume that
the ranking~`$\mathord{\rk}$' is implicit.  Given a set~$T$ and a
ranked alphabet~$\Sigma$, we let
\[ \Sigma(T) = \{ \sigma(\seq t1n) \mid n \in \nat, \sigma \in
\Sigma_n, \seq t1n \in T\} \enspace. \] The set~$T_\Sigma(Q)$ of
\emph{$\Sigma$-trees indexed by a set~$Q$} is the smallest set~$T$
such that $Q \cup \Sigma(T) \subseteq T$.  We write~$T_\Sigma$
for~$T_\Sigma(\emptyset)$.  Given a tree~$t \in T_\Sigma(Q)$, its
positions~$\pos(t) \subseteq \nat^*$ are inductively defined by
$\pos(q) = \{\varepsilon\}$ for each $q \in Q$ and $\pos(\sigma(\seq
t1n)) = \{\varepsilon\} \cup \{iw \mid i \in [n], w \in \pos(t_i) \}$
for all $n \in \nat$, $\sigma \in \Sigma_n$, and $\seq t1n \in
T_\Sigma(Q)$.  For each position~$w \in \pos(t)$, we write~$t(w)$ for
the label of~$t$ at position~$w$ and $t|_w$~for the subtree of~$t$
rooted at~$w$.  Formally,
\begin{align*}
  q(\varepsilon) &= q & \bigl( \sigma(\seq t1n) \bigr)(w) &=
  \begin{cases}
    \sigma & \text{if } w = \varepsilon \\
    t_i(v) & \text{if } w = iv \text{ with } i \in [n],\, v \in \nat^*
  \end{cases} \\
  q|_\varepsilon &= q & \sigma(\seq t1n)|_w &=
  \begin{cases}
    \sigma(\seq t1n) & \text{if } w = \varepsilon \\
    t_i|_v & \text{if } w = iv \text{ with } i \in [n],\, v \in \nat^*
  \end{cases}
\end{align*}
for all $q \in Q$, $n \in \nat$, $\sigma \in \Sigma_n$, and $\seq t1n
\in T_\Sigma(Q)$.  The height~$\height(t)$ of a tree~$t \in
T_\Sigma(Q)$ is simply $\height(t) = \max\ \{\abs w \mid w \in
\pos(t)\}$.

We reserve the use of the special symbol~$\SBox$ of rank~$0$.  A tree
$t \in T_{\Sigma \cup \{\SBox\}}(Q)$ is a \emph{$\Sigma$-context
  indexed by~$Q$} if the symbol~$\SBox$ occurs exactly once in~$t$.
The set of all $\Sigma$-contexts indexed by~$Q$ is denoted
by~$C_\Sigma(Q)$.  As before, we write~$C_\Sigma$
for~$C_\Sigma(\emptyset)$.  For each $c \in C_\Sigma(Q)$ and $t \in
T_\Sigma(Q)$, the substitution~$c[t]$ denotes the tree obtained
from~$c$ by replacing~$\SBox$ by~$t$.  Similarly, we use the
substitution~$c[c']$ with another context~$c' \in C_\Sigma(Q)$, in
which case we obtain yet another context.

We take all weights from a \emph{commutative semifield} $\langle \sr
S, \mathord\srplus, \mathord\srtimes, 0, 1\rangle$,\footnote{We
  generally require $0 \neq 1$, and in fact, the additive monoid is
  rather irrelevant for our purposes.} which is an algebraic structure
consisting of a commutative monoid~$\langle \sr S, \mathord{\srplus},
0\rangle$ and a commutative group $\langle \sr S - \{0\},
\mathord{\srtimes}, 1\rangle$ such that
\begin{compactitem}
\item $s \srtimes 0 = 0$ for all $s \in \sr S$, and
\item $s \srtimes (s_1 \srplus s_2) = (s \srtimes s_1) \srplus (s
  \srtimes s_2)$ for all $s, s_1, s_2 \in \sr S$.
\end{compactitem}
Roughly speaking, commutative semifields are commutative
semirings~\cite{hebwei98,gol99} with multiplicative inverses.  Many
practically relevant weight structures are commutative semifields.
Examples include
\begin{compactitem}
\item the real numbers $\langle \sr R, \mathord{+}, \mathord{\cdot},
  0, 1\rangle$,
\item the tropical semifield $\langle \sr R \cup \{\infty\},
  \mathord{\min}, \mathord{+}, \infty, 0\rangle$,
\item the probabilistic semifield $\langle [0, 1], \mathord{\max},
  \mathord{\cdot}, 0, 1\rangle$ with $[0, 1] = \{ r \in \sr R \mid 0
  \leq r \leq 1\}$, and
\item the \textsc{Boolean} semifield $\sr B = \langle \{0, 1\},
  \mathord{\max}, \mathord{\min}, 0, 1\rangle$.
\end{compactitem}
For the rest of the paper, let $\langle \sr S, \mathord\srplus,
\mathord\srtimes, 0, 1\rangle$ be a commutative semifield (with $0
\neq 1$), and let $\underline{\sr S} = \sr S - \{0\}$.  For every $s
\in \underline{\sr S}$ we write $s^{-1}$ for the inverse of~$s$; i.e.,
$s \srtimes s^{-1} = 1$.  For better readability, we will sometimes
write $\tfrac{s_1}{s_2}$ instead of $s_1 \srtimes s_2^{-1}$.  The
following notions implicitly use the commutative semifield~$\sr S$.  A
\emph{weighted tree language} is simply a mapping $\varphi \colon
T_\Sigma(Q) \to \sr S$.  Its \emph{support}~$\supp(\varphi) \subseteq
T_\Sigma(Q)$ is $\supp(\varphi) = \varphi^{-1}(\underline{\sr S})$;
i.e., the support contains exactly those trees that are evaluated to
non-zero by~$\varphi$.  Given $s \in \sr S$, we let $(s \cdot \varphi)
\colon T_\Sigma(Q) \to \sr S$ be the weighted tree language such that
$(s \srtimes \varphi)(t) = s \srtimes \varphi(t)$ for every $t \in
T_\Sigma(Q)$.

A deterministic weighted tree automaton
(\wdta)~\cite{berreu82,kui98,bor05a,fulvog09} is a tuple ${\cal A} =
(Q, \Sigma, \delta, \mathord{\wt}, F)$ with
\begin{compactitem}
\item a finite set~$Q$ of states,
\item a ranked alphabet~$\Sigma$ of input symbols such that $\Sigma
  \cap Q = \emptyset$,
\item a transition mapping $\delta \colon \Sigma(Q) \to Q$,\footnote{Note
  that our \wdta\ are always total.  We additionally disallow
  transition weight~$0$.  If a transition is undesired, then its
  transition target can be set to a sink state, which we commonly
  denote by~$\bot$.  Finally, the restriction to final states instead
  of final weights does not cause a difference in expressive power in
  our setting~\protect{\cite[Lemma~6.1.4]{bor05a}}.} 
\item a transition weight assignment $\mathord{\wt} \colon \Sigma(Q)
  \to \underline{\sr S}$, and
\item a set~$F \subseteq Q$ of final states.
\end{compactitem}
The transition and transition weight mappings
`$\delta$'~and~`$\mathord{\wt}$' naturally extend to mappings
$\hat\delta \colon T_\Sigma(Q) \to Q$ and $\hat{\mathord{\wt}} \colon
T_\Sigma(Q) \to \underline{\sr S}$ by
\begin{align*}
  \hat\delta(q) &= q & \hat\delta(\sigma(\seq t1n)) &=
  \delta(\sigma(\hat\delta(t_1), \dotsc, \hat\delta(t_n))) \\*
  \hat{\wt}(q) &= 1 & \hat{\wt}(\sigma(\seq t1n)) &=
  \wt(\sigma(\hat\delta(t_1), \dotsc, \hat\delta(t_n))) \srtimes
  \prod_{i \in [n]} \hat{\wt}(t_i)
\end{align*}
for every $q \in Q$, $n \in \nat$, $\sigma \in \Sigma_n$, and $\seq
t1n \in T_\Sigma(Q)$.  Since $\hat{\delta}(t) = \delta(t)$ and
$\hat{\wt}(t) = \wt(t)$ for all $t \in \Sigma(Q)$, we can safely omit
the hat and simply write $\delta$~and~$\wt$ for
$\hat\delta$~and~$\hat{\wt}$, respectively.  The \wdta~${\cal A}$
recognizes the weighted tree language $\mathord{\sem{\mathcal A}}
\colon T_\Sigma \to \sr S$ such that
\[ \sem{\mathcal A}(t) =
\begin{cases}
  \wt(t) & \text{if } \delta(t) \in F \\
  0 & \text{otherwise}
\end{cases} \] for all $t \in T_\Sigma$.  Two \wdta\ $\mathcal
A$~and~$\mathcal B$ are equivalent if $\sem{\mathcal A} =
\sem{\mathcal B}$; i.e., their recognized weighted tree languages
coincide.  A \wdta\ over the \textsc{Boolean} semifield~$\sr B$ is
also called \dta~\cite{gecste84,gecste97} and written $(Q, \Sigma,
\delta, F)$ since the component~`$\wt$' is uniquely determined.
Moreover, we identify each \textsc{Boolean} weighted tree language
$\varphi \colon T_\Sigma(Q) \to \{0,1\}$ with its support.  Finally,
the set~$C_\delta$ of shallow transition contexts is
\[ C_\delta = \{ \sigma(\seq q1{i-1}, \SBox, \seq q{i+1}n) \mid n \in
\nat, i \in [n], \sigma \in \Sigma_n, \seq q1n \in Q \} \enspace, \]
which we assume to be totally ordered by some arbitrary order~$\leq$.

For minimization, the weighted (extended) context language of a state
is relevant.  For every $q \in Q$ the context-semantics~$\sem q_{{\cal
    A}} \colon C_\Sigma(Q) \to \sr S$ of~$q$ is defined for every $c
\in C_\Sigma(Q)$ by
\[ \semf q_{\mathcal A}(c) =
\begin{cases}
  \wt(c[q]) & \text{if } \delta(c[q]) \in F \\
  0 & \text{otherwise.}
\end{cases} \] Intuitively, $\sem q_{\mathcal A}$~is the weighted
(extended) language recognized by~${\cal A}$ starting in state~$q$.
Two states $q, q' \in Q$ are equivalent~\cite{bor03}, written $q
\equiv q'$, if there exists $s \in \underline{\sr S}$ such that $\semf
q_{{\cal A}}(c) = s \cdot \semf{q'}_{{\cal A}}(c)$ for all $c \in
C_\Sigma$.  An equivalence relation~$\mathord{\cong} \subseteq Q
\times Q$ is a congruence relation (for the \wdta~$\mathcal A$) if for
all $n \in \nat$, $\sigma \in \Sigma_n$, and $q_1 \cong q'_1, \dotsc,
q_n \cong q'_n$ we have $\delta(\sigma(\seq q1n)) \cong
\delta(\sigma(\seq{q'}1n))$.  It is known~\cite{bor03}
that~$\mathord{\equiv}$ is a congruence relation.  The \wdta~${\cal
  A}$ is minimal if there is no equivalent \wdta\ with strictly fewer
states.  We can compute a minimal \wdta\ efficiently using a variant
of \textsc{Hopcroft}'s algorithm~\cite{hop71,hogmalmay08} that
computes~$\mathord{\equiv}$ and runs in time~$\bigo(m \log n)$, where
$m = \abs{\Sigma(Q)}$~is the size of~$\mathcal A$ and $n = \abs Q$.

\section{A characterization of hyper-minimality}
\label{sec:Hyper}
Hyper-minimization~\cite{badgefshi07} is a form of lossy compression
that allows any finite number of errors.  It has been investigated
in~\cite{bad08,holmal10,gawjez09} for deterministic finite-state
automata and in~\cite{jezmal12b} for deterministic tree automata.
Finally, hyper-minimization was already generalized to weighted
deterministic finite-state automata in~\cite{malque11b}, from which we
borrow much of the general approach.  In the
following, let ${\cal A} = (Q, \Sigma, \delta, \mathord{\wt}, F)$ and
${\cal B} = (P, \Sigma, \mu, \mathord{\wt'}, G)$ be \wdta\ over the
commutative semifield $\langle \sr S, \mathord\srplus,
\mathord\srtimes, 0, 1 \rangle$ with $0 \neq 1$.

We start with the basic definition of when two weighted tree languages
are almost-equivalent.  We decided to use the same approach as
in~\cite{malque11b}, so we require that the weighted tree languages,
seen as functions, must coincide on almost all trees.  Note that this
restriction is not simply the same as requiring that the weighted tree
languages have almost-equal (i.e., finite-difference) supports.  It
fact, our definition yields that the supports are almost-equal, but
that is not sufficient.  In addition, we immediately allow a scaling
factor in many of our basic definitions since those are already
required in classical minimization~\cite{bor03} to obtain the most
general statements.  Naturally, a scaling factor is not allowed for
the almost-equivalence of \wdta\ since these are indeed supposed
assign a different weight to only finitely many trees.

\begin{definition}
  \label{df:AE}
  Two weighted tree languages $\varphi_1, \varphi_2 \colon T_\Sigma(Q)
  \to \sr S$ are \emph{almost-equivalent}, written $\varphi_1 \approx
  \varphi_2$, if there exists $s \in \underline{\sr S}$ such that
  $\varphi_1(t) = s \cdot \varphi_2(t)$ for almost all $t \in
  T_\Sigma(Q)$.\footnote{``Almost all'' means all but a finite number,
    as usual.}  We write $\varphi_1 \approx \varphi_2 \pod s$ to
  indicate the factor~$s$.  The \wdta\ ${\cal A}$~and~${\cal B}$ are
  almost-equivalent if $\sem{{\cal A}} \approx \sem{{\cal B}} \pod 1$.
  Finally, the states $q \in Q$ and $p \in P$ are almost-equivalent if
  there exists $s \in \underline{\sr S}$ such that $\semf q_{\mathcal
    A}(c) = s \srtimes \semf p_{\mathcal B}(c)$ for almost all $c \in
  C_\Sigma$.
\end{definition}

We start with some basic properties of~$\approx$, which is shown to be
an equivalence relation both on the weighted tree languages as well as
on the states of a single \wdta.  In addition, we demonstrate that the
latter version is even a congruence relation.  This shows that once we
are in almost-equivalent states, the same impetus causes the different
devices to switch to other almost-equivalent states.  

\begin{lemma}
  \label{lem.cong}
  Almost-equivalence is an equivalence relation such that
  $\delta(c[q]) \approx \mu(c[p])$ for all $c \in C_\Sigma$, $q \in
  Q$, and $p \in P$ with $q \approx p$.
\end{lemma}

\begin{proof}
  Trivially, $\approx$~is reflexive and symmetric (because we have
  multiplicative inverses for all elements of~$\underline{\sr S}$).
  Let $\varphi_1, \varphi_2, \varphi_3 \colon T_\Sigma(Q) \to \sr S$
  be weighted tree languages such that $\varphi_1 \approx \varphi_2
  \pod s$ and $\varphi_2 \approx \varphi_3 \pod{s'}$ for some $s, s'
  \in \underline{\sr S}$.  Then there exist finite sets $L, L'
  \subseteq T_\Sigma(Q)$ such that $\varphi_1(t) = s \srtimes
  \varphi_2(t)$ and $\varphi_2(t') = s' \srtimes \varphi_3(t')$ for
  all $t \in T_\Sigma(Q) - L$ and $t' \in T_\Sigma(Q) - L'$.
  Consequently, $\varphi_1(t'') = s \srtimes s' \srtimes
  \varphi_3(t'')$ for all $t'' \in T_\Sigma(Q) - (L \cup L')$, which
  proves $\varphi_1 \approx \varphi_3 \pod{s \srtimes s'}$ and thus
  transitivity.  Hence, $\approx$~is an equivalence relation.  The
  same arguments can be used for~$\approx$ on \wdta\footnote{Note that
    $1^{-1} = 1$ and $1 \cdot 1 = 1$, so the restriction to factor~$1$
    in the definition of the almost-equivalence of \wdta\ is not
    problematic.} and states.  For the second property, induction
  allows us to easily prove~\cite{bor05a} that
  \begin{align*}
    \semf q_{{\cal A}}(c'[c]) &= \wt(c[q]) \srtimes
    \semf{\delta(c[q])}_{{\cal A}}(c') \quad &\text{and}&& \quad
    \semf p_{{\cal B}}(c_2[c_1]) &= \wt'(c_1[p]) \srtimes
    \semf{\mu(c_1[p])}_{{\cal B}}(c_2)
    \tag{$\dagger$} 
  \end{align*}
  for all $c, c' \in C_\Sigma(Q)$ and $c_1,c_2 \in C_\Sigma(P)$.
  Since $q \approx p \pod s$, there exists a finite set~$C \subseteq
  C_\Sigma$ such that $\semf q_{{\cal A}}(c'') = s \srtimes \semf
  p_{{\cal B}}(c'')$ for all $c'' \in C_\Sigma - C$.  Consequently,
  \[ \semf{\delta(c[q])}_{{\cal A}}(c') = \frac{\semf q_{{\cal
        A}}(c'[c])} {\wt(c[q])} = s \srtimes \frac{\semf p_{{\cal
        B}}(c'[c])} {\wt(c[q])} = s \srtimes
  \frac{\wt'(c[p])}{\wt(c[q])} \srtimes \semf{\mu(c[p])}_{{\cal
      B}}(c') \] for all~$c' \in C_\Sigma$ such that $c'[c] \notin C$,
  which proves that $\delta(c[q]) \approx \mu(c[p])$.
\end{proof}

Next, we show that almost-equivalent states of the same \wdta\ even
coincide (up to the factor~$s$) on almost all extended contexts, which
are contexts in which states may occur.

\begin{lemma}
  \label{lem.extended}
  Let $\mathcal A$ be minimal and $q \approx q' \pod s$ for some $s \in
  \underline{\sr S}$ and $q, q' \in Q$.  Then
  \[ \sem q\nolimits_{\mathcal A}(c) = s \srtimes \sem
  {q'}\nolimits_{\mathcal A}(c) \tag{$\ddagger$} \]
  for almost all $c \in C_\Sigma(Q)$.
\end{lemma}

\begin{proof}
  By definition of~$q \approx q' \pod s$, there exists a finite set~$C
  \subseteq C_\Sigma$ such that $\sem q_{\mathcal A}(c) = s \srtimes
  \sem{q'}_{\mathcal A}(c)$ for all $c \in C_\Sigma - C$.  Let $h \geq
  \max\ \{\height(c) \mid c \in C\}$ be an upper bound for the height
  of those finitely many contexts.  Clearly, there are only finitely
  many contexts of~$C_\Sigma$ that have height at most~$h$.  Now, let
  $c \in C_\Sigma(Q)$ be an extended context such that $\height(c) >
  h$, and let $W = \{ w \in \pos(c) \mid c(w) \in Q\}$ be the
  positions that are labeled with states.  For each state~$q \in Q$,
  select $t_q \in \delta^{-1}(q) \cap T_\Sigma$ a tree (without
  occurrences of states) that is processed in~$q$.  Clearly, such a
  tree exists for each state because $\mathcal A$~is minimal.  Let
  $c'$~be the context obtained from~$c$ by replacing each state
  occurrence of~$q$ by~$t_q$.  Obviously, $\height(c') \geq
  \height(c) > h$ because we replace only leaves.  Consequently, $c' \in
  C_\Sigma - C$.  Using a variant~\cite{bor05a} of~$(\dagger)$ we
  obtain
  \[ \semf q_{\mathcal A}(c) \cdot \prod_{w \in W} \wt(t_{c(w)}) =
  \semf q_{\mathcal A}(c') = s \srtimes \semf{q'}_{\mathcal A}(c') = s
  \srtimes \semf {q'}_{\mathcal A}(c) \cdot \prod_{w \in W}
  \wt(t_{c(w)}) \enspace, \] where the second equality is due to the
  fact that $c' \in C_\Sigma - C$.  Comparing the left-hand and
  right-hand side and cancelling the additional terms, which is
  allowed in a commutative semifield, we obtain $\semf q_{\mathcal
    A}(c) = s \srtimes \semf{q'}_{\mathcal A}(c)$ for all $c \in
  C_\Sigma(Q)$ with $\height(c) > h$, and thus for almost all $c \in
  C_\Sigma(Q)$ as required.
\end{proof}

As in all the other scenarios, the goal of hyper-minimization given
device~$\mathcal A$ is to construct an almost-equivalent device~${\cal
  B}$ such that no device is smaller than~${\cal B}$ and
almost-equivalent to~${\cal A}$.  In our setting, the devices are
\wdta\ over the ranked alphabet~$\Sigma$ and the commutative
semifield~$\sr S$.  Since almost-equivalence is an equivalence
relation by Lemma~\ref{lem.cong}, we can replace the requirement
``almost-equivalent to~${\cal A}$'' by ``almost-equivalent to~${\cal
  B}$'' and call a \wdta~${\cal B}$ \emph{hyper-minimal} if no
(strictly) smaller \wdta\ is almost-equivalent to it.  Then
hyper-minimization equates to the computation of a hyper-minimal
\wdta~${\cal B}$ that is almost-equivalent to~$\cal A$.  Let us first
investigate hyper-minimality, which was characterized
in~\cite{badgefshi07} for the \textsc{Boolean} semifield using the
additional notion of a \emph{preamble} state.

\begin{definition}[see \protect{\cite[Definition~2.11]{badgefshi07}}]
  \label{df:Kernel}
  A state~$q \in Q$ is a \emph{preamble state} if $\delta^{-1}(q) \cap
  T_\Sigma$ is finite.  Otherwise, it is a \emph{kernel state}.
\end{definition}

In other words, a state is a preamble state if and only if it accepts
finitely many trees (without occurrences of states).  This notion is
essentially unweighted, so the discussion in~\cite{jezmal12b} applies.
In particular, we can compute the set of kernel states in
time~$\mathcal O(m)$ with $m = \abs{\Sigma(Q)}$ being the size of the
\wdta~$\mathcal A$.

Recall that a \wdta\ (without unreachable states; i.e.,
$\delta^{-1}(q) \cap T_\Sigma \neq \emptyset$ for every $q \in Q$) is
minimal if and only if it does not have a pair of different, but
equivalent states~\cite{bra68,bor03}.  The ``only-if'' part of this
statement is shown by merging two equivalent states to obtain a
smaller, but equivalent \wdta.  Let us define a merge that
additionally applies a weight~$s$ to the rerouted transitions.

\begin{definition}
  \label{df:Merge}
  Let $q, q' \in Q$ and $s \in \underline{\sr S}$ with $q \neq q'$.
  The \emph{$s$-weighted merge} of~$q$ into~$q'$ is the\\
  \wdta~$\merge_{\cal A}(q \stackrel s\to q') = (Q - \{q\}, \Sigma,
  \delta', \mathord{\wt'}, F - \{q\})$ such that for all $t \in
  \Sigma(Q - \{q\})$
  \begin{align*}
    \delta'(t) &=
    \begin{cases}
      q' & \text{if } \delta(t) = q \\
      \delta(t) & \text{otherwise}
    \end{cases} &&& \wt'(t) &=
    \begin{cases}
      s \srtimes \wt(t) & \text{if } \delta(t) = q \\
      \wt(t) & \text{otherwise.}
    \end{cases}
  \end{align*}
\end{definition}

In our approach to weighted hyper-minimization, we also merge, but we
need to take care of the factors, so we use the weighted merges just
introduced.  The next lemma hints at the correct use of weighted
merges.

\begin{lemma}
  \label{lem.merge}
  Let $q, q' \in Q$ be different states, of which $q$~is a preamble
  state, and $s \in \underline{\sr S}$ be such that $q \approx q' \pod
  s$.  Then $\merge_{{\cal A}}(q \stackrel s\to q')$ is
  almost-equivalent to~${\cal A}$.
\end{lemma}

\begin{proof}
  Since $q \approx q' \pod s$, there exists a finite set~$C \subseteq
  C_\Sigma$ such that $\semf q_{\mathcal A}(c) = s \srtimes
  \semf{q'}_{\mathcal A}(c)$ for all $c \in C_\Sigma - C$.  Let $h
  \geq \max\ \{\height(c) \mid c \in C\}$~be an upper bound on the
  height of the contexts of~$C$.  Moreover, let $h' \geq \max\
  \{\height(t) \mid t \in \delta^{-1}(q) \cap T_\Sigma\}$~be an upper
  bound for the height of the trees of~$\delta^{-1}(q) \cap T_\Sigma$,
  which is a finite set since $q$~is a preamble state.  Finally, let
  $z > h + h'$.  Now we return to the main claim.  Let ${\cal B} =
  \merge_{\mathcal A}(q \stackrel s\to q')$ and consider an arbitrary
  tree~$t \in T_\Sigma$ whose height is at least~$z$.  Clearly,
  showing that $\mathcal B(t) = \mathcal A(t)$ for all trees~$t$ with
  $\height(t) \geq z$ proves that $\mathcal B$~and~$\mathcal A$ are
  almost-equivalent.\footnote{There are only finitely many ranked
    trees up to a certain height and recall that almost-equivalence
    does not permit a scaling factor for \wdta.}  Let $W = \{ w \in
  \pos(t) \mid \delta(t|_w) = q\}$ be the set of positions of the
  subtrees that are recognized in state~$q$.  Now $\wt'(t|_w) = s
  \srtimes \wt(t|_w)$ for all $w \in W$ because clearly the
  subtrees~$t|_w$ only use states different from~$q$ except at the
  root, where $\mathcal A$~switches to~$q$ and $\mathcal B$~switches
  to~$q'$ with the additional weight~$s$.  Note that $q$~cannot occur
  anywhere else inside those subtrees because this would create a loop
  which is impossible for a preamble state.  Let $W = \{\seq w1m\}$
  with $w_1 \sqsubset \dotsm \sqsubset w_m$, in which $\sqsubseteq$~is
  the lexicographic order on~$\nat^*$.  Let $c_1 \in C_\Sigma$ be the
  context obtained by removing the subtree at~$w_1$ from~$t$.  Note
  that $c_1$~is taller than~$h$ (i.e., $\height(c_1) > h$) and thus
  $c_1 \in C_\Sigma - C$ because the height of~$t$ is larger than~$h +
  h'$ and the height of~$t|_{w_1}$ is at most~$h'$.  Consequently,
  using a variant~\cite{bor05a} of~$(\dagger)$ we obtain
  \begin{align*}
    \mathcal A(t) &= \mathcal A(c_1[t|_{w_1}]) \stackrel{\dagger}=
    \wt(t|_{w_1}) \srtimes \semf q_{\mathcal A}(c_1) =
    \frac{\wt'(t|_{w_1})}{s} \srtimes s \srtimes \semf{q'}_{\mathcal
      A}(c_1) = \wt'(t|_{w_1})
    \srtimes \semf{q'}_{\mathcal A}(c_1) \\
    &= \wt'(t|_{w_1}) \srtimes
    \begin{cases}
      \wt(c_1[q']) & \text{if } \delta(c_1[q']) \in F \\
      0 & \text{otherwise.}
    \end{cases} \\
    \intertext{Let $c_2$ be the context obtained from~$c_1[q']$ by
      replacing the subtree at~$w_2$ by~$\SBox$.  Also $c_2 \notin C$.}
    &= \wt'(t|_{w_1}) \srtimes
    \begin{cases}
      \wt(c_2[t|_{w_2}]) & \text{if } \delta(c_2[t|_{w_2}]) \in F \\
      0 & \text{otherwise.}
    \end{cases} \quad = \wt'(t|_{w_1}) \srtimes \wt(t|_{w_2})
    \srtimes \semf q_{\mathcal A}(c_2) \\
    &\stackrel{\ddagger}= \wt'(t|_{w_1}) \srtimes
    \frac{\wt'(t|_{w_2})}{s} \srtimes s \srtimes \semf {q'}_{\mathcal
      A}(c_2) = \wt'(t|_{w_1}) \srtimes \wt'(t|_{w_2}) \srtimes \semf
    {q'}_{\mathcal A}(c_2) \enspace, \\
    \intertext{which can now be iterated to obtain} &= \wt'(t|_{w_1})
    \srtimes \ldots \srtimes \wt'(t|_{w_m}) \srtimes \semf
    {q'}_{\mathcal A}(c_m) = \wt'(t|_{w_1}) \srtimes \ldots \srtimes
    \wt'(t|_{w_m}) \srtimes \semf {q'}_{\mathcal B}(c_m)
    \stackrel{\dagger}= \mathcal B(t) \enspace,
  \end{align*}
  where the second-to-last step is justified because the state~$q$ is
  not used when processing the context~$c_m$.  This proves the
  statement.
\end{proof}

\begin{theorem}
  \label{thm:HMChar}
  A minimal \wdta\ is hyper-minimal if and only if it has no pair of
  different, but almost-equivalent states, of which at least one is a
  preamble state.
\end{theorem}

\begin{proof}
  Let ${\cal A}$ be the minimal \wdta.  For the ``only if'' part, we
  know by Lemma~\ref{lem.merge} that the smaller \wdta\
  $\merge_{\mathcal A}(q \stackrel s\to q')$ is almost-equivalent
  to~${\cal A}$ if $q \approx q' \pod s$ and $q$~is a preamble state.
  For the ``if'' direction, suppose that ${\cal B}$ is
  almost-equivalent to~${\cal A}$ and $\abs P < \abs
  Q$.\footnote{Recall that almost-equivalent \wdta\ do not permit a
    scaling factor; their semantics need to coincide for almost all
    trees.}  For all $t \in T_\Sigma$ we have $\delta(t) \approx
  \mu(t)$ by Lemma~\ref{lem.cong}.  Since $\abs P < \abs Q$, there
  exist $t_1, t_2 \in T_\Sigma$ with $q_1 = \delta(t_1) \neq
  \delta(t_2) = q_2$ but $\mu(t_1) = p = \mu(t_2)$.  Consequently,
  $q_1 = \delta(t_1) \approx \mu(t_1) = p = \mu(t_2) \approx
  \delta(t_2) = q_2$, which yields $q_1 \approx q_2$.  By assumption,
  $q_1$~and~$q_2$ are kernel states.  Using a variation of the above
  argument (see~\cite[Theorem~3.3]{badgefshi07}) we can obtain
  $t_1$~and~$t_2$ with the above properties such that $\height(t_1),
  \height(t_2) \geq \abs Q^2$.  Due to their heights, we can pump the
  trees $t_1$~and~$t_2$, which yields that the states $\langle q_1,
  p\rangle$~and~$\langle q_2, p\rangle$ are kernel states of the
  \textsc{Hadamard} product $\mathcal A \srtimes \mathcal B$.  Since
  ${\cal A}$~and~${\cal B}$ are almost-equivalent, we have
  \begin{alignat*}{3}
    \wt(t_1) \srtimes \semf{q_1}_{{\cal A}}(c) &\stackrel{\dagger}=
    \semf{{\cal A}}(c[t_1]) &&= \semf{{\cal B}}(c[t_1])
    &&\stackrel{\dagger}= \wt'(t_1) \srtimes \semf p_{{\cal B}}(c) \\
    \wt(t_2) \srtimes \semf{q_2}_{{\cal A}}(c) &\stackrel{\dagger}=
    \semf{{\cal A}}(c[t_2]) &&= \semf{{\cal B}}(c[t_2])
    &&\stackrel{\dagger}= \wt'(t_2) \srtimes \sem p\nolimits_{{\cal
        B}}(c)
  \end{alignat*}
  for almost all $c \in C_\Sigma$ using again the tree variant
  of~($\dagger$).  Moreover, since both $\langle q_1,
  p\rangle$~and~$\langle q_2, p\rangle$ are kernel states, we can
  select $t_1$~and~$t_2$ such that the previous statements are
  actually true for all $c \in C_\Sigma$.  Consequently,
  \begin{align*}
    \frac{\wt(t_1) \srtimes \semf{q_1}_{{\cal A}}(c)} {\wt'(t_1)} =
    \frac{\wt(t_2) \srtimes \semf{q_2}_{{\cal A}}(c)} {\wt'(t_2)}
    \quad \text{and} \quad \semf{q_1}_{{\cal A}}(c) = s \srtimes
    \semf{q_2}_{{\cal A}}(c)
  \end{align*}
  for all $c \in C_\Sigma$ and $s = \frac{\wt'(t_1) \srtimes \wt(t_2)}
  {\wt'(t_2) \srtimes \wt(t_1)}$, which yields $q_1 \equiv q_2$.  This
  contradicts minimality since $q_1 \neq q_2$, which shows that such a
  \wdta~${\cal B}$ cannot exist.
\end{proof}

\section{Hyper-minimization}
\label{sec:HM}
Next, we consider some algorithmic aspects of hyper-minimization for
\wdta.  Since the unweighted case is already well-described in the
literature~\cite{jezmal12b}, we focus on the weighted case, for which
we need the additional notion of co-preamble states~\cite{malque11b},
which in analogy to~\cite{malque11b} are those states with finite
support of their weighted context language.  Let $P$~and~$K$ be the
sets of preamble and kernel states of~${\cal A}$, respectively.

\begin{definition}
  \label{def.cokernel}
  A state~$q \in Q$ is a \emph{co-preamble state} if $\supp(\sem
  q_{{\cal A}})$~is finite.  Otherwise it is a \emph{co-kernel
    state}.  The sets of all co-preamble states and all co-kernel
  states are $\copreamble$~and~$\cokernel = Q - \copreamble$,
  respectively.
\end{definition}

Transitions entering a co-preamble state can be ignored while checking
almost-equivalence because (up to a finite number of weight
differences) the reached states behave like the sink state~$\bot$.
Trivially, all co-preamble states are almost-equivalent.  In addition,
a co-preamble state cannot be almost-equivalent to a co-kernel state.
The interesting part of the almost-equivalence is thus completely
determined by the weighted languages of the co-kernel states.  This
special role of the co-preamble states has already been pointed out
in~\cite{gawjez09} in the context of \textsc{dfa}.

All hyper-minimization
algorithms~\cite{badgefshi07,bad09,gawjez09,holmal10} share the same
overall structure (Algorithm~\ref{alg:Overall}).  In the final step we
perform state merges (see Definition~\ref{df:Merge}).  Merging only
preamble states into almost-equivalent states makes sure that the
resulting \wdta\ is almost-equivalent to the input \wdta\ by
Lemma~\ref{lem.merge}.  Algorithm~\ref{alg:Overall} first minimizes
the input \wdta\ using, for example, the algorithm
of~\cite{malque11c}.  With the help of a weight redistribution along
the transitions (pushing), it reduces the problem to \dta\
minimization, for which we can use a variant of \textsc{Hopcroft}'s
algorithm~\cite{hogmalmay08}.  In the next step, we compute the
set~$K$ of kernel states of~${\cal A}$~\cite{jezmal12b} using any
algorithm that computes strongly connected components (for example,
\textsc{Tarjan}'s algorithm~\cite{tar72}).  By~\cite{jezmal12b} a
state is a kernel state if and only if it is reachable from (i)~a
nontrivial strongly connected component or (ii)~a state with a
self-loop.  Essentially, the same approach can be used to compute the
co-kernel states.  In line~\ref{ln.ae} we compute the
almost-equivalence on the states~$Q$, which is the part where the
algorithms~\cite{badgefshi07,bad09,gawjez09,holmal10} differ.
Finally, we merge almost-equivalent states according to
Lemma~\ref{lem.merge} until the obtained \wdta\ is hyper-minimal (see
Theorem~\ref{thm:HMChar}).

\renewcommand{\algorithmicensure}{\textbf{Return:}}
\begin{algorithm}[t]
  \begin{algorithmic}[2]
    \REQUIRE a \wdta~${\cal A}$ with $n$~states
    \ENSURE an almost-equivalent hyper-minimal \wdta
    \smallskip \hrule \smallskip
    \STATE ${\cal A} \gets \textsc{Minimize}({\cal A})$
    \label{ln.o1}
      \COMMENT{$\bigo(m \log n)$} 
    \STATE $K \gets \textsc{ComputeKernel}({\cal A})$
      \COMMENT{$\bigo(m)$}
    \STATE $\overline K \gets \textsc{ComputeCoKernel}({\cal A})$
      \COMMENT{$\bigo(m)$}
    \STATE $(\mathord\sim, t) \gets
    \textsc{ComputeAlmostEquivalence}({\cal A}, \overline K)$
    \label{ln.ae}
      \COMMENT{Algorithm~\protect{\ref{alg:Almost}} --- $\bigo(m \log n)$}
    \RETURN $\textsc{MergeStates}({\cal A}, K, \mathord\sim, t)$
      \COMMENT{Algorithm~\protect{\ref{alg.merge}} --- $\bigo(m)$}
  \end{algorithmic}
  \caption{Structure of the hyper-minimization algorithm.} 
  \label{alg:Overall}
\end{algorithm}

\begin{lemma}
  \label{lm.AE}
  Let ${\cal A}$ be a minimal \wdta.  The states $q, q' \in Q$ are
  almost-equivalent if and only if there is $n \in \nat$ such that
  $\delta(c[q]) = \delta(c[q'])$ for all $c \in C_\Sigma$ such that
  $\SBox$~occurs at position~$w$ in~$c$ with $\abs w \geq n$.
\end{lemma}

Our algorithm for computing the almost-equivalence is an extension of
the algorithm of~\cite{malque11b}.  As in~\cite{malque11b}, we need to
handle the scaling factors, for which we introduced the standardized
signature in~\cite{malque11b}.  Roughly speaking, we ignore
transitions into co-preamble states and normalize the transition
weights.  Recall that~$C_\delta$ is the set of transition contexts;
i.e., transitions with exactly one occurrence of the symbol~$\SBox$.
Moreover, for every $q \in Q$, we let $c_q$~be the smallest transition
context~$c_q \in C_\delta$ such that $\delta(c_q[q]) \in \overline K$,
where the total order on~$C_\delta$ is arbitrary as assumed earlier,
but it needs to be consistently used.

\begin{definition}
  \label{def.signature}
  Given $q \in Q$, its \emph{standardized signature} is
  \[ \sg(q) = \Bigl\{ \langle c, \delta(c[q]),
  \frac{\wt(c[q])}{\wt(c_q[q])} \rangle \;\Bigl|\; c \in C_\delta,\,
  \delta(c[q]) \in \cokernel \Bigr\} \enspace. \]
\end{definition}

Next, we show that states with equal standardized signature are indeed
almost-equivalent.

\begin{lemma}
  \label{lem.signature}
  For all $q, q' \in Q$, if $\sg(q) = \sg(q')$, then $q \approx q'$.
\end{lemma}

\begin{proof}
  If $q$~or~$q'$ is a co-preamble state, then both $q$~and~$q'$ are
  co-preamble states and thus $q \approx q'$.  Now, let $q, q' \in
  \cokernel$, and let $c_q \in C_\delta$ be the smallest transition
  context such that $c_q[q] \in \cokernel$.  Since $q'$~has the same
  signature, $c_q = c_{q'}$.  In addition, let $s = \frac{\wt(c_q[q])}
  {\wt(c_q[q'])}$.  For every $c \in C_\delta$ and $c' \in C_\Sigma$,
  \[ \semf q_{{\cal A}}(c'[c]) \stackrel{\dagger}= \wt(c[q]) \srtimes
  \semf{\delta(c[q])}_{{\cal A}}(c') \quad \text{and} \quad
  \semf{q'}_{{\cal A}}(c'[c]) \stackrel{\dagger}= \wt(c[q']) \srtimes
  \semf{\delta(c[q'])}_{{\cal A}}(c') \enspace. \] 

  First, let $\langle c, q_c, s_c \rangle \notin \sg(q) = \sg(q')$ for
  all $q_c \in Q$ and $s_c \in \underline{\sr S}$.  Then $c$~takes
  both $q$~and~$q'$ into a co-preamble state and thus $\semf q_{{\cal
      A}}(c'[c]) = 0 = s \srtimes \sem{q'}_{{\cal A}}(c'[c])$ for
  almost all $c' \in C_\Sigma$.  Second, suppose that $\langle c, q_c,
  s_c \rangle \in \sg(q) = \sg(q')$ for some $q_c \in Q$ and $s_c \in
  \underline{\sr S}$.  Since $\delta(c[q]) = q_c = \delta(c[q'])$, and
  we obtain
  \begin{alignat*}{3}
    \semf q_{{\cal A}}(c'[c]) &= \frac{\wt(c[q])}{\wt(c_q[q])}
    \srtimes \wt(c_q[q]) \srtimes \semf {q_c}_{{\cal A}}(c') &&= s_c
    \srtimes \wt(c_q[q])
    \srtimes \semf {q_c}_{{\cal A}}(c') \\*
    &= \frac{\wt(c[q'])} {\wt(c_q[q'])} \srtimes \wt(c_q[q]) \srtimes
    \semf{q_c}_{{\cal A}}(c') &&= s \srtimes \semf{q'}_{{\cal
        A}}(c'[c])
  \end{alignat*}
  for every $c' \in C_\Sigma$, which shows that $q \approx q' \pod s$
  because the scaling factor~$s$ does not depend on the
  transition context~$c$.
\end{proof}

In fact, the previous proof can also be used to show that at most the
empty context~$\SBox$ yields a difference in the weighted context
languages $\sem q_{{\cal A}}$~and~$\sem{q'}_{{\cal A}}$ (up to the
common factor).  For the completeness, we also need a (restricted)
converse for minimal \wdta, which shows that as long as there are
almost-equivalent states, we can also identify them using the
standardized signature.

\begin{algorithm}[t!]
  \begin{algorithmic}[2]
    \REQUIRE minimal \wdta~${\cal A}$ and its co-kernel
    states~$\overline K$
    \ENSURE almost-equivalence~$\approx$ as a partition and scaling
    map~$f \colon Q \to \underline{\sr K}$ \smallskip
    \hrule
    \smallskip\smallskip
    \FORALL{$q \in Q$}
      \STATE $\pi(q) \gets \{q\}$; $f(q) \gets 1$ 
        \COMMENT{trivial initial blocks}
    \ENDFOR
    \smallskip
    \STATE $h \gets \emptyset$; $I \gets Q$
      \COMMENT{hash map of type~$h \colon \sg \to Q$}
    \smallskip
    \FORALL{$q \in I$}
      \STATE $\suc \gets \sg(q)$	
        \COMMENT{compute standardized signature using current~$\delta$
        and $\overline K$}%
      \smallskip
      \IF{$\textsc{HasValue}(h, \suc)$}
        \STATE $q' \gets \textsc{Get}(h, \suc)$
          \COMMENT{retrieve state in bucket~`$\suc$' of~$h$}%
          \smallskip
        \IF{$\abs{\pi(q')} \geq \abs{\pi(q)}$}
          \STATE $\textsc{Swap}(q, q')$
            \label{ln.Comp}
            \COMMENT{exchange roles of $q$~and~$q'$}%
          \smallskip
        \ENDIF
        \STATE $I \gets I \cup \{r \in Q - \{q'\} \mid \exists
          c \in C_\delta \colon \delta(c[r]) = q'\}$ 
          \COMMENT{add predecessors of~$q'$}
        \STATE $f(q') \gets \frac{\wt(c_q[q'])}{\wt(c_q[q])}$
          \COMMENT{$c_q$ is as in Definition~\ref{def.signature}}
        \STATE ${\cal A} \gets \merge_{{\cal A}}(q' \stackrel{f(q')}\to q)$ 
          \label{ln.Merge}
          \COMMENT{merge~$q'$ into~$q$}
        \STATE $\pi(q) \gets \pi(q) \cup \pi(q')$
          \COMMENT{$q$~and~$q'$ are almost-equivalent}%
        \smallskip
        \FORALL{$r \in \pi(q')$}
          \STATE $f(r) \gets f(r) \cdot f(q')$
          \COMMENT{recompute scaling factors}
        \ENDFOR
        \smallskip
      \ENDIF
      \STATE $h \gets \textsc{Put}(h, \suc, q)$
        \COMMENT{store~$q$ in~$h$ under key~`$\suc$'}
    \ENDFOR
    \smallskip
    \RETURN $(\pi, f)$ 
  \end{algorithmic}
  \caption{Algorithm computing the almost-equivalence~$\approx$ and
    scaling map~$f$.}
  \label{alg:Almost}
\end{algorithm}

\begin{lemma}
  \label{lem.join}
  Let ${\cal A}$ be minimal, and let $q \approx q'$ be such that
  $\sg(q) \neq \sg(q')$.  Then there exist $r, r' \in Q$ such that $r
  \neq r'$ and $\sg(r) = \sg(r')$.
\end{lemma}

\begin{proof}
  Since $q \approx q'$, there exists an integer~$h$ such that
  $\delta(c[q]) = \delta(c[q'])$ for all $c \in C_\Sigma$ such that $w
  \in \pos(c)$ with $c(w) = \SBox$ and $\abs w \geq h$ by
  Lemma~\ref{lm.AE}.  Let~$c' \in C_\Sigma$ be a maximal context such
  that $r = \delta(c'[q]) \neq \delta(c'[q']) = r'$.  Since $c'$ is
  maximal, we have $\delta(c''[c'[q]]) = q_{c''} =
  \delta(c''[c'[q']])$ for all $c'' \in C_\delta$.  If $q_{c''}$~is a
  co-preamble state, then $\langle c, q_c, s_c \rangle \notin \sg(r) =
  \sg(r')$ for all $q_c \in Q$ and $s_c \in \underline{\sr S}$.  On
  the other hand, let $q_{c''}$~be a co-kernel state, and let $c_r \in
  C_\delta$ be the smallest transition context such that
  $\delta(c_r[r]) \in \cokernel$.  Since $q \approx q'$ and
  $\mathord{\approx}$ is a congruence relation by
  Lemma~\ref{lem.cong}, we have $r \approx r' \pod s$ for some $s \in
  \underline{\sr S}$, which means that $\semf r_{{\cal A}}(c) = s
  \srtimes \semf{r'}_{{\cal A}}(c)$ for almost all $c \in C_\Sigma$.
  Consequently,
  \begin{align*}
    \wt(c''[r]) \srtimes \semf{q_{c''}}_{{\cal A}}(c) &= s
    \srtimes \wt(c''[r']) \srtimes \semf {q_{c''}}_{{\cal A}}(c) \\*
    \wt(c_r[r]) \srtimes \semf{\delta(c_r[r])}_{{\cal A}}(c) &= s
    \srtimes \wt(c_r[r']) \srtimes \semf{\delta(c_r[r])}_{{\cal A}}(c)
  \end{align*}
  for almost all $c \in C_\Sigma$.  Since both
  $q_{c''}$~and~$\delta(c_r[r])$ are co-kernel states, we immediately
  can conclude that $\wt(c''[r]) = s \srtimes \wt(c''[r'])$ and
  $\wt(c_r[r]) = s \srtimes \wt(c_r[r'])$, which yields
  \[ \frac{\wt(c''[r])}{\wt(c_r[r])} = \frac{s \srtimes
    \wt(c''[r'])} {s \srtimes \wt(c_r[r'])} =
  \frac{\wt(c''[r'])}{\wt(c_r[r'])} \enspace. \]
  This proves $\sg(r) = \sg(r')$ as required.
\end{proof}

Lemmata \ref{lem.signature}~and~\ref{lem.join} suggest
Algorithm~\ref{alg:Almost} for computing the almost-equivalence and a
map representing the scaling factors.  This map contains a scaling
factor for each state with respect to a representative state of its
block.  Algorithm~\ref{alg:Almost} is a straightforward modification
of an algorithm by~\cite{holmal10} using our standardized signatures.
We first compute the standardized signature for each state and store
it into a (perfect) hash map~\cite{diekarmehmeyrohtar94} to avoid
pairwise comparisons.  If we find a collision (i.e., a pair of states
with the same signature), then we merge them such that the state
representing the bigger block survives (see Lines
\ref{ln.Comp}~and~\ref{ln.Merge}).  Each state is considered at
most~$\log n$ times because the size of the ``losing'' block
containing it at least doubles.  After each merge, scaling factors of
the ``losing'' block are computed with respect to the new
representative.  Again, we only recompute the scaling factor of
each state at most~$\log n$ times.
Hence the small modifications compared to~\cite{holmal10} do not
increase the asymptotic run-time of Algorithm~\ref{alg:Almost}, which
is $\bigo(n \log n)$ where $n$~is the number of states (see Theorem~9
in~\cite{holmal10}).  Alternatively, we can use the standard reduction
to a weighted finite-state automaton using each transition context~$c
\in C_\delta$ as a new symbol.

\begin{proposition}
  \label{lm:Time}
  Algorithm~\ref{alg:Almost} can be implemented to run in
  time~$\bigo(m \log n)$, where $m = \abs{\Sigma(Q)}$ and $n = \abs Q$.
\end{proposition}

Finally, we need an adjusted merging process that takes the scaling
factors into account.  When merging one state into another, their mutual
scaling factor can be computed from the scaling map by multiplicaton of
one scaling factor with the inverse of the other.
Therefore, merging (see Algorithm~\ref{alg.merge}) can be implemented
in time $\bigo(n)$, and hyper-minimization
(Algorithm~\ref{alg:Overall}) can be implemented in time $\bigo(m \log
n)$ in the weighted setting.

\begin{algorithm}[t]
  \begin{algorithmic}[2]
    \REQUIRE a minimal \wdta~${\cal A}$, its kernel states~$K$,
      its almost-equivalence~$\mathord\approx$, and a scaling map~$f
      \colon Q \to \underline{\sr S}$
    \ENSURE hyper-minimal \wdta~${\cal A}$ that is almost-equivalent
    to the input \wdta \smallskip
    \hrule
    \smallskip
    \FORALL{$B \in (Q/\mathord\approx)$}
      \STATE select $q \in B$ such that $q \in K$ if possible
      \FORALL{$q' \in B - K$}
        \STATE ${\cal A} \gets \merge_{\cal A}(q'
        \stackrel{\frac{f(q')}{f(q)}}\longrightarrow q)$
      \ENDFOR
    \ENDFOR
  \end{algorithmic}
  \caption{Merging almost-equivalent states.}
  \label{alg.merge}
\end{algorithm}

\begin{proposition}
  \label{lm:Time2}
  Our hyper-minimization algorithm can be implemented to run in
  time~$\bigo(m \log n)$.
\end{proposition}

It remains to prove the correctness of our algorithm.  To prove the
correctness of Algorithm~\ref{alg:Almost}, we still need a technical
property.

\begin{lemma}
  \label{lem.techmerge}
  Let $q, q' \in Q$ be states with $q \neq q'$ but $\sg(q) = \sg(q')$.
  Moreover, let ${\cal B} = \merge_{\mathcal A}(q' \stackrel s\to q)$
  with $s = \frac{f(q')}{f(q)}$, and let $\cong$ be its
  almost-equivalence (restricted to~$P$).  Then $\mathord{\cong} =
  \mathord{\approx} \cap (P \times P)$ where $P = Q - \{q'\}$.
\end{lemma}

\begin{proof}
  Let $p_1 \approx p_2$ with $p_1, p_2 \in P$.  Let $c =
  c_\ell[c_{\ell-1}[\cdots [c_1] \cdots]]$ with $\seq c1\ell \in
  C_\delta$.  Then we obtain the runs
  \begin{align*}
    R_{p_1} &= \tup{\delta(c_1[p_1]), \delta(c_2[c_1[p_1]]), \dotsm,
      \delta(c[p_1])} \quad \text{with weight $\wt(c[p_1])$} \\
    R_{p_2} &= \tup{\delta(c_1[p_2]), \delta(c_2[c_1[p_2]]), \dotsm,
      \delta(c[p_2])} \quad \text{with weight $\wt(c[p_2])$.}
  \end{align*}
  The corresponding runs $R'_{p_1}$~and~$R'_{p_2}$ in~${\cal B}$
  replace every occurrence of~$q'$ in both $R_{p_1}$~and~$R_{p_2}$
  by~$q$.  Their weights are
  \begin{align*}
    \wt'(c[p_1]) &= \begin{cases}
      \wt(c[p_1]) & \text{if } \delta(c[p_1]) \neq q' \\ 
      \wt(c[p_1]) \srtimes s & \text{otherwise}
    \end{cases} \\
    \wt'(c[p_2]) &= \begin{cases} 
      \wt(c[p_2]) & \text{if } \delta(c[p_2]) \neq q' \\ 
      \wt(c[p_2]) \srtimes s & \text{otherwise.} \end{cases} 
  \end{align*}
  Since $\delta(c'[p_1]) = \delta(c'[p_2])$ for suitably tall
  contexts~$c' \in C_\Sigma$ and $p_1 \approx p_2$, we obtain that
  $p_1 \cong p_2$.  The same reasoning can be used to prove the
  converse.
\end{proof}

\begin{theorem}
  \label{thm.main}
  Algorithm~\ref{alg:Almost} computes~$\approx$ and a scaling map.
\end{theorem}

\begin{proof}[Proof sketch]
  If there exist different, but almost-equivalent states,
  then there exist different states with the same standardized
  signature by Lemma~\ref{lem.join}.  Lemma~\ref{lem.signature} shows
  that such states are almost-equivalent.\linebreak[4]  Finally,
  Lemma~\ref{lem.techmerge} shows that we can continue the computation
  of the almost-equivalence after a weighted merge of such states.
  The correctness of the scaling map is shown implicitly in the
  proof of Lemma~\ref{lem.signature}.
\end{proof}

\begin{theorem}
  \label{thm:HMWTime}
  We can hyper-minimize \wdta\ in time~$\bigo(m \log n)$, where $m =
  \abs{\Sigma(Q)}$ and $n = \abs Q$.
\end{theorem}


\bibliographystyle{eptcs}
\bibliography{extra}

\end{document}